\documentclass[journal]{IEEEtran}
\usepackage{multicol}
\usepackage{amsmath, amsthm, amssymb}
\usepackage{longtable,supertabular}
\usepackage{amsmath,algorithm,algpseudocode}
\usepackage{amssymb}
\usepackage{latexsym}
\usepackage{mathrsfs}

\newtheorem{theorem} {{Theorem}}
\newtheorem{proposition} {{Proposition}}

\newcommand{\cP}{{\cal P}}

\newcommand{\sP}{\cP}
\newcommand{\Ps}{\smash{{\sP\kern-2.0pt}_q\kern-0.5pt(n)}}

\newcommand{\Gr}{\mathcal{G}_q(k,n)}

\newcommand{\GF}{\mathbb{F}}

\DeclareMathOperator{\rank}{rank}
\DeclareMathOperator{\im}{im}

\title{Improving the Linkage Construction with Echelon-Ferrers for Constant-Dimension Codes}

\author{Xianmang He, Yindong Chen, Zusheng Zhang}

\begin{document}
\maketitle
\begin{abstract}
Echelon-Ferrers  is an important method to improve lower bounds for constant-dimension codes, which can be applied on various parameters.  Fagang Li~\cite{FagangLi2019Construction} combined the linkage construction and echelon-Ferrers to obtain some new lower bounds of constant-dimension codes.
In this letter, we generalize this linkage construction to obtain new lower bounds.

\textbf{keywords:} Subspace Coding,  Linkage Construction, Echelon-Ferrers Construction, Constant-Dimension Codes
\end{abstract}

\section{Introduction}

Let $\GF_q$ be the finite field with $q$ elements.
The set of all $k$-dimensional subspaces of an $\GF_q$-vector space $V$ will be denoted by $\Gr$.
In general, the {\it projective space} of order $n$ over the finite field~\smash{$\GF_q$}, denoted by $\Ps$, is the set of all subspaces of the vector space~\smash{$\GF_q^n$}.
All these subspaces form a metric space with respect to the \emph{subspace distance}, which is defined as
\[\begin{split}
  d_S(U,W):=&\   \dim(U+W)-\dim(U\cap W)\\
           =&\   2\cdot\dim(U+W)-\dim(U)-\dim(W),
\end{split}\]
where $U$ and $W$ are subspaces of  $\GF_q^n$.

A set $\mathcal{C}$ of subspaces of $V$ is called a \emph{subspace code}. The \emph{minimum distance} of $\mathcal{C}$ is given by $d = \min\{d_S(U,W) \mid U,W\in\mathcal{C},\ U \neq W\}$.  If the dimension of the \emph{codewords} is fixed as $k$, we use the notation $(n,\#\mathcal{C},d, k)_q$ and call $\mathcal{C}$ a \emph{constant dimension code} (CDC for short).  The maximal possible size of an $(n, M, d,k)_q$ CDC is often denoted by $A_q(n,d,k)$.

Subspace coding was first proposed by R. K\"oetter and F. R. Kschischang   in \cite{koetter2008coding}  to error control in random network coding.
The main problem in subspace coding is to determine the maximal possible size ${\bf A}_q(n,d,k)$,  which makes the subspace distance satisfies:
for any two different subspaces $U$ and $W$, we have $d(U,W)=2k-2\dim(U \cap W) \geq d$.

A plethora of  results on the construction of CDCs are invented in the literatures.  The lower and upper bounds on $A_q(n,d,k)$ have been in-depth studied in the last decade, see \cite{Heinlein2016Tables}.
The report \cite{Heinlein2016Tables}  describes  an on-line database, which we refer to the online tables http://subspacecodes.uni-bayreuth.de. These tables gather up-to-date information about the current lower and upper bounds for subspace codes.
Lifted \emph{maximum rank-distance} (MRD for short) codes are one type of building blocks of the echelon-Ferrers construction \cite{etzion2009error}. The idea of multilevel
construction is widely used, including parallel construction~\cite{XuChen2018}, coset construction \cite{Heinlein2015Coset}, pending dot~\cite{silberstein2014subspace},  etc.
The most powerful construction is the linkage construction \cite{Gluesing2017linkage} and its improved construction~\cite{heinlein2017asymptotic}. The linkage construction is improved by these new works \cite{Antonio2019Combining,He2019Construction,HeinleinGeneralized}.

Recently, Fagang Li combined the two methods of linkage construction and echelon-Ferrers to obtain some new lower bounds of CDCs.
In this letter we generalize this construction. Based on this, we further improved the construction by a greedy algorithm.

\section{the Combing Method in \cite{FagangLi2019Construction}}

Let $X$ be a $k$-dimensional subspace of $\Gr$. We can represent $X$ by the matrix in reduced
row echelon form $E(X)$, whose $k$ rows form a basis for $X$.
The \emph{identifying vector} of $X$, denoted by $v(X)$, is a binary vector of length $n$ and weight $k$, where the  $k$ \emph{ones} of $v(X)$ are exactly the pivots of $E(X)$.

The zeroes are removed from each row of $E(X)$, which lie on the left of the pivots. Then we delete the columns which exactly having the pivots.
After that, all the remaining entries are shifted to the right.  We finally get the Ferrers tableaux form of a subspace $X$, denoted by $\mathcal{F}(X)$.
The Ferrers diagram of $X$ can be obtained from $\mathcal{F}(X)$ by replacing the entries of $\mathcal{F}(X)$ with dots.

Let $F_q^{m\times \ell}$ be an $m\times \ell$  matrices space over the field $F_q$.
For any two distinct matrices $A,B\in F_q^{m\times \ell}$, the rank-metric is defined as
$d_R(A, B) := \rank(A-B)$.  A subset of $F_q^{m\times \ell}$ with the rank-metric is called a rank-metric code.
If a rank-metric code is a linear subspace of $F_q^{m\times \ell}$,  we can call it a linear rank-metric code.
It is clear that the rank-distance of a rank-metric code $\mathcal{C}$ can be defined as
$d_R(\mathcal{C}) := \min\{d_R(A,B): A,B \in \mathcal{C},A \neq B\}$.
It is well-known that the number of codewords in  $\mathcal{C}$ is upper bounded by
$q^{\max\{m,\ell\}\cdot(\min\{m,\ell\}-d+1)}$. A code attaining this bound is called
a maximum rank-distance (MRD) code.

Let $\mathcal{F}$ be a Ferrers diagram with $\ell$ dots in the top row and $m$ dots in the rightmost
column. If for any codeword $M$ of $\mathcal{C_F}$, all entries of $M$ not in $\mathcal{F}$ are zeroes,
a linear rank-metric code $\mathcal{C_F}$ of $F^{m\times \ell }_q$ is called a Ferrers diagram rank-metric (FDRM)
code. An FDRM code $\mathcal{C_F}$ is denoted an $[F,d, \delta]$
FDRM code, if $\rank(A)\ge d$ for any nonzero codeword $A$,
and $\dim(\mathcal{C_F}) = \delta$.

The following theorem determines an upper bound on the
size of $\dim(\mathcal{C_F})$.

\begin{theorem} (see \cite{etzion2009error}) \label{thm-upper-bound-ef}
  Let $\mathcal{F}$ be the Ferrers diagram of $\ell$ in the top row and $m$ dots in the rightmost column.
  Let $\mathcal{C_\mathcal{F}} \subseteq F_q^{m\times \ell}$  be the corresponding FDRM code fulfilling $\forall A, B\in \mathcal{C_\mathcal{F}}$, $\rank(A-B)\ge \delta$.
  Then $|\mathcal{C_\mathcal{F}}| \le q^{\min_i\{w_i\}}$, where $w_i$ is the number of dots in $\mathcal{F}$,
  which are neither contained in the rightmost $\delta-1-i$ columns nor contained  in the first $i$ rows  for $0\le i\le \delta-1$.
\end{theorem}

Furthermore, the authors of \cite{etzion2009error} proved that the upper bound can be attained when $d=2,4$,
and then conjectured that the upper bound is also tight for other cases.

For simplify, for any given  matrix $M\in F^{k\times \ell}_q$ over $F_q$, the row space of $M$ is denoted by $\im(M)$.

We recall some basic notations of linkage in \cite{Gluesing2017linkage}.
A set $\mathcal{U} \subset F_q^{k\times n}$ with the size $k \times n$ matrices over $F_q$ is called an SC-representation set
if $\rank(U)=k$ for all $U \in \mathcal{U}$ and $\im(U_1) \neq \im(U_2)$ for all $U_1 \neq U_2$ in $\mathcal{U}$.

\begin{proposition}(see \cite{Gluesing2017linkage}).
Let $U$ be an SC-representation  set of a $(n_1, N_1, d_1, k)_q$ constant dimension subspace code
and $\mathcal{M} \subset F_q^{k \times n_2}$  be a linear rank-metric code with distance $d_2$ and $N_2$ elements.
Consider the set of $k$ dimension subspaces in ${\bf F}_q^{n_1+n_2}$ defined by ${C}=\{\im(U|M): U \in  \mathcal{U}, M \in  \mathcal{M}\}$.
This is an $(n_1+n_2, N_1N_2, \min\{d_1,2d_2\}, k)_q$ constant dimension code. Here $(U|M)$ is a $k \times (n_1+n_2)$ matrix concatenated from $U$ and $Q$.
\end{proposition}

We quote the following theorem (Theorem 3.1 in \cite{FagangLi2019Construction}) to briefly describe the construction method,  see the paper \cite{FagangLi2019Construction} for details.

\begin{theorem} \label{them-combining} \cite{FagangLi2019Construction}
Let $n_1 > k,\ n_2 > k,\ k\ge d$.
For $i=1,2$,
let $\mathcal{U}_i \subseteq F_q^{k\times n_i}$ be SC-representing sets with cardinality $N_i$, and $d_S(\mathcal{U}_i)=d$.
Assume that $\mathcal{C}_R  \subseteq F^{k\times n_2}_q$ is a linear rank-metric code with  $|\mathcal{C}_R| = N_R$ and $d_R(\mathcal{C}_R) = \frac{d}{2}$.

Let the identifying
vectors $v_j$ with length $n := n_1 + n_2$ and weight $k$ satisfy the
following properties for $j=1,2,\cdots$.

(a) For each identifying vector  $v_j$, the count of $1$'s in the first $n_1$ positions and the last $n_2$ positions are both greater than or equal to $\frac{d}{2}$.

(b) For any two distinct identifying vectors $v_{j_1}$ and $v_{j_2}$,  the Hamming distance $H(v_{j_1},v_{j_2})\ge d.$

Let $\mathcal{C}_{\mathcal{F}_j} \subseteq F^{k\times (n-k)}_q$ be an FDRM code and $d_R(\mathcal{C}_{\mathcal{F}_j}) = \frac{d}{2}$,
where $\mathcal{F}_j$ is a Ferrers diagram corresponding to the identifying vector $v_j$.

Denote by $C$ the subspace code of length $n = n_1 + n_2$ as $C = C_1 \cup C_2 \cup C_3$,
where

$C_1 = \{\im(U|M)\ |\ U \in \mathcal{U}_1, M\in \mathcal{C}_R\}$;

$C_2 = \{\im(\mathbf{0}_{k\times n_1}|U) \ |\ U \in \mathcal{U}_2\}$;

$C_3=\cup_j \mathbb{C}_{\mathcal{F}_j}$, $\mathbb{C}_{\mathcal{F}_j}$ is the lifted FDRM code of $\mathcal{C}_{\mathcal{F}_j}$.

Thus, $C$ is an $(n,N, d, k)_q$ CDC with $N = N_2 + N_1\cdot N_R + \sum_j |\mathbb{C}_{\mathcal{F}_j}|$.

\end{theorem}

This construction modifies the echelon-Ferrers construction, which replaces the lifted MRD code $\im(I_k, C_{R_1})$ with the linkage construction $\im(U, C_{R_2})$,
where $C_{R_1}$ and $C_{R_2}$ are linear rank-metric codes with size $k\times (n-k)$ and $k\times n_2$, respectively.

\section{Construction And Algorithm}
In this section, we're going to give the details of our construction and an algorithm with greedy strategy.

\subsection{General Construction}
We now generalize the multilevel construction and linkage construction.

\begin{theorem}\label{them-general}
Let $n_1 \ge k,\ n_2 \ge k,\ \mathcal{U}_1 \subseteq F_q^{k\times n_1}$ be SC-representing sets with cardinality $N_1$, and $d_S(\mathcal{U}_1)=d$.
Assume that $\mathcal{C}_R  \subseteq F^{k\times n_2}_q$ is a linear rank-metric code with $|\mathcal{C}_R| = N_R$ and $d_R(\mathcal{C}_R) = \frac{d}{2}$.
Let the identifying vector $v_j$ with length $n := n_1 + n_2$ and weight $k$ satisfy the
following properties for $j=1,2,\cdots$.

(a) For any identifying vector $v_j$, the count of $1$'s in the last $n_2$ positions is at least $\frac{d}{2}$.

(b)  For any two distinct identifying vectors $v_{j_1}$ and $v_{j_2}$,  the Hamming distance $H(v_{j_1},v_{j_2})\ge d.$

Let $\mathcal{C}_{\mathcal{F}_j} \subseteq F^{k\times (n-k)}_q$ be an FDRM code and $d_R(\mathcal{C}_{\mathcal{F}_j}) = \frac{d}{2}$,
where $\mathcal{F}_j$ is a Ferrers diagram corresponding to the identifying vector $v_j$.

Denote by $C$ the subspace code of length $n = n_1 + n_2$ as $C = C_1 \cup C_2 $,
where

$C_1 = \{\im(U|M)\ |\ U \in \mathcal{U}_1, M\in \mathcal{C}_R\}$;

$C_2=\cup_j \mathbb{C}_{\mathcal{F}_j}$, $\mathbb{C}_{\mathcal{F}_j}$ is the lifted FDRM code of $\mathcal{C}_{\mathcal{F}_j}$.

Thus, $C$ is an $(n,N,d,k)_q$ CDC with $N = N_1\cdot N_R + \sum_j |\mathbb{C}_{\mathcal{F}_j}|$.

\end{theorem}

\begin{proof}
We note that $C_1$ is an $(n_1+n_2, N_1N_R, d, k)_q$ constant dimension code, and $C_2$ is the set of the lifted FDRM code.
Consider that the pivots of  $C_1$ and  $C_2$ are pairwise disjoint, therefore, the cardinality of the code is  $N_1\cdot N_R + \sum_j |\mathbb{C}_{\mathcal{F}_j}|$.

According to the definition, $\mathbb{C}_{\mathcal{F}_j}$ is a CDC with $d_S(C_{\mathcal{F}_j})\ge d$.
Hence, it is sufficient to prove that for any $w_1\in C_1$, $w_2\in C_2$, $d_S(w_1, w_2)\ge d$.
In light of the definition of  subspace distance $d_S$ as mentioned before, it is equivalent to prove  that $\dim(\im(w_1)+ \im(w_2)) \ge k+\frac{d}{2}$.

For any identifying vector $v_j\ (j=1,2,\cdots)$,  we note that this vector has $\frac{d}{2}$ ones in the last $n_2$ positions,
and can be illustrated  in reduced row echelon form as follows:

$$Z_1:=
\left(
    \begin{array}{cccccccc}
      Z_{11} & Z_{12}  \\
      Z_{21} & Z_{22}  \\
      \end{array}
  \right)_{k\times n},
$$

where $Z_{11}$ is a matrix with the size of  $(k-\frac{d}{2})\times n_1$, $Z_{12}$ is a matrix with the size of $(k-\frac{d}{2})\times n_2$, $Z_{21}$ is a zero matrix with the size of ${\frac{d}{2}\times n_1}$,  $Z_{22}$ is a matrix with the size of  $\frac{d}{2}\times n_2$, and $Z_{22}$ contains at least $\frac{d}{2}$ pivots.

It is clear that  $\im(w_1)=\im(U|M)$, $U\in \mathcal{U}_1, M\in \mathcal{M}$, then
$$
Z_2:=\left(
    \begin{array}{cccccccc}
      U_{k\times n_1}  & M_{k\times n_2} \\
      Z_{11} & Z_{12}  \\
      Z_{21} & Z_{22}  \\
      \end{array}
  \right)_{2k\times n}.
$$

Notice that the rank of $U$ is $k$, $\rank(Z_{22}) \ge \frac{d}{2}$, and $Z_{21}= \mathbf{0}_{\frac{d}{2}\times n_2}$.
Hence, the rank of the  matrix $Z_2$ is at least $k+\frac{d}{2}$.
Here we finish the proof of the theorem.

\end{proof}

Remark: Under careful comparison, we can find that the construction in Theorem \ref{them-general} differs from the one in Theorem \ref{them-combining}
in that the restriction of $k\ge d$ is removed, and the condition ($a$) is relaxed,  sacrificing  the code $C_2$  in Theorem \ref{them-combining}.
The size of code $C_2$ contains only small codewords, and the candidate set of identifying vectors increases greatly, which will finally add more code into $\cup_j \mathbb{C}_{\mathcal{F}_j}$.  These findings are verified by the greedy algorithm.

\subsection{Greedy Algorithm}
Due to the limitation of data scale, the optimal echelon-Ferrers is difficult to operate effectively.
Therefore, we employ an algorithm with greedy strategy, which is illustrated detailedly in Algorithm~\ref{alg:Framwork}.
\begin{algorithm}[htb]
\caption{Greedy()\label{alg:Framwork}}
\begin{algorithmic}[1]
\Require
$n_1,\ n_2,\ d,\ k$
\Ensure
target identifying vector set $S_v$
\State construct $V_{set}$: a number of $\sum_{\Delta=0}^{k-\frac{d}{2}}\binom{n_1}{k-\frac{d}{2}-\Delta}\times \binom{n_2}{\frac{d}{2}+\Delta}$ identifying vectors
\State compute corresponding dimensions for each vector in $V_{set}$
\State sort $V_{set}$ in descending order by the dimension values
\State pick up the first vector of the sorted $V_{set}$ and put it to $S_v$
\For {$i= \text{maxdimension}-1$ down to $0$}
\State  $iSet$: compatible vectors with dimension $i$ in $V_{set}$
\State  choose $v$ from $iSet$ under the greedy criteria:
        it has minimum distance to the latest vector in $S_v$
\State  pick $v$ out from $iSet$ and put it to $S_v$
\State  repeat Step $7$ and $8$ until there's no more such $v$
\EndFor
\end{algorithmic}
\end{algorithm}

The greedy algorithm operates by selecting identifying vectors and adding them to the target set $S_v$.
Firstly, a total number of $\sum_{\Delta=0}^{k-\frac{d}{2}}\binom{n_1}{k-\frac{d}{2}-\Delta}\times \binom{n_2}{\frac{d}{2}+\Delta}$ identifying vectors are added to $V_{set}$.
For all the vectors in $V_{set}$, we compute their corresponding dimensions by Theorem~\ref{thm-upper-bound-ef}, and sort them in descending order according to the value of dimensions.
The target set $S_v$ is empty initially, and the first vector (with maximum dimension) of the sorted $V_{set}$ is put into $S_v$.
Then the loop step runs from the second maximum dimension down to dimension $0$.
For each round with dimension $i$,
we denote by $iSet$ as the set of vectors with dimension $i$ in $V_{set}$ and compatible to $S_v$,
i.e., $$iSet = \{ v\in V_{set}\ |\ \dim(v)=i,\ d_H(v, s_j)\ge d,\ \forall \ s_j \in S_v\}.$$
Now, we select vectors from $iSet$ and add them to $S_v$ one by one.
In order to add vectors as more as possible,
a greedy strategy is employed:
the vector $v$ has the minimum Hamming distance to the latest vector in $S_v$.
Then vector $v$ is picked out from $iSet$ and added to $S_v$.
This process will continue until no more such vector $v$ can be found to add to $S_v$.
%
%
%
In some case of $i$, maybe the $iSet$ is an empty set.
The total cost of the algorithm is bounded by $O(m\cdot\log m)$, where $m$ equals $\sum_{\Delta=0}^{k-\frac{d}{2}}\binom{n_1}{k-\frac{d}{2}-\Delta}\times \binom{n_2}{\frac{d}{2}+\Delta}$.

~\\
\noindent \textbf{Example $1 \quad$}
In order to apply Theorem~\ref{them-general} for $A_q(12, 4, 4)$, we can choose  $n_1=8$ and $n_2=4$.
By applying the echelon-Ferrers construction, a number of $25$ identifying vectors are obtained.
We list all the obtained identifying vectors in descending order according to their dimension values in Table~\ref{tab-12-4-4}.
It is known that $A_q(8, 4, 4) \ge  q^{12} + q^2(q^2 + 1)2(q^2 + q + 1) + 1$.
Then we have
$A_q(12,4,4)\ge q^{12}(q^{12} +(q^2 +q +1)(q^2 +1)^2(q^4 +1)+(q^{12}+2q^{10}+3q^9+5q^8+q^7+2q^6+q^5+7q^4+q^3+q^2+1).$
When $q = 2$, we obtain $A_2(12,4,4)\ge 19674269$, which is an improvement of the corresponding results in \cite{FagangLi2019Construction,XuChen2018}.
However, it is still weaker than the paper \cite{Antonio2019Combining,He2019Construction,HeinleinGeneralized}.

\begin{table}[tb]
\caption{ Construction for $A_q(12,4,4)$}\label{tab-12-4-4}
\begin{tabular}{|c|c|c||c|c|c|}\hline
  & Identifying Vector  &  Dim  &  & Identifying Vector  &  Dim  \\ \hline
  \hline
1 & 110000001100 & 12 & 14 & 001100000011 & 6  \\ \hline
2 & 101000001010 & 10 & 15 & 000001101001 & 5  \\ \hline
3 & 001100001100 & 10 & 16 & 000010011001 & 4  \\ \hline
4 & 011000001001 & 9 & 17 & 000000111100 & 4  \\ \hline
5 & 011000000110 & 9 & 18 & 000011000011 & 4  \\ \hline
6 & 010100001010 & 9 & 19 & 000010010110 & 4  \\ \hline
7 & 110000000011 & 8 & 20 & 000010100101 & 4  \\ \hline
8 & 101000000101 & 8 & 21 & 000001100110 & 4  \\ \hline
9 & 100100001001 & 8 & 22 & 000001011010 & 4  \\ \hline
10 & 100100000110 & 8 & 23 & 000001010101 & 3  \\ \hline
11 & 000011001100 & 8 & 24 & 000000110011 & 2  \\ \hline
12 & 010100000101 & 7 & 25 & 000000001111 & 0  \\ \hline
13 & 000010101010 & 6 &  &  &   \\ \hline
\end{tabular}
\end{table}


\subsection{Examples}

In this section, we give several examples constructed by our methods, and in the meanwhile the expressions of these bounds are also given.


\begin{table*}[htbp]
\centering
\caption{\label{tab-n-4-4}New subspace codes on $A_q(n,4,4)$}
\begin{tabular}{|c|c|c|}\hline
    ${\bf A}_q(n,d,k)$　&　New 　&　Old \\ \hline
  \hline
$A_2(13,4,4)$ &157396313    & 157332190   \\ \hline
$A_3(13,4,4)$ &7793514240823    &  7793495430036  \\ \hline
$A_4(13,4,4)$ &18118665490931521   &               18118664249474716  \\ \hline
$A_5(13,4,4)$ &7466568820575245751   &          7466568787180077320   \\ \hline
$A_7(13,4,4)$ &65745512221518213208951   &   65745512216555289614188  \\ \hline
$A_8(13,4,4)$ &2418546150658513179095553  & 2418546150622126921477496  \\ \hline
$A_9(13,4,4)$ &58159941504105053602711351 & 58159941503893673245551936   \\ \hline
$A_2(14,4,4)$ & 1259180741   &1258757174   \\ \hline
$A_3(14,4,4)$ & 210424885305967   &210424421624298    \\ \hline
$A_4(14,4,4)$ & 1159594591440676369   &               1159594516050838620  \\ \hline
$A_5(14,4,4)$ & 933321102572187066901   &          933321098538702991570   \\ \hline
$A_7(14,4,4)$ & 22550710691980761970230475   &   22550710690309028764671498  \\ \hline
$A_8(14,4,4)$ & 1238295629137158820126564417  & 1238295629118788686643907448  \\ \hline
$A_9(14,4,4)$ & 42398597356492584370649345569&42398597356340204444957848530   \\ \hline
$A_2(15,4,4)$ & 10073479745    & 10071464646   \\ \hline
$A_3(15,4,4)$ & 5681471907063670    & 5681463153275925   \\ \hline
$A_4(15,4,4)$ & 74214053852327765337   &               74214050169101548368  \\ \hline
$A_5(15,4,4)$ & 116665137821525349488286   &          116665137415279661027650   \\ \hline
$A_7(15,4,4)$ & 7734893767349401489942798302   &   7734893766857015258769289566  \\ \hline
$A_8(15,4,4)$ & 634007362118225316632582459985  & 634007362109986775858834010688  \\ \hline
$A_9(15,4,4)$ & 30908577472883094009455515497142&  30908577472784286989399940957138   \\ \hline
\end{tabular}
\end{table*}


1) $d\ge k$

Let $n_1=8,\ n_2=5$, apply the algorithm, we have $A_q(13,4,4)\ge A_q(8,4,4)\times  q^{15}+q^{15}+2q^{13}+3q^{12}+5q^{11}+q^{10}+3q^9+6q^8+7q^7+5q^6+3q^5+3q^4+q^3+1$.
When $q=2$, there's $A_2(13,4,4)\ge 157396313$.  This bound is strictly improves upon the corresponding results in  \cite{Antonio2019Combining,ChenHeWengXu2020,heinlein2017asymptotic,XuChen2018,HeinleinGeneralized,He2019Construction}.

Let $n_1=8,\ n_2=6$, we have $A_q(14,4,4)\ge A_q(8,4,4)\times  q^{18}+q^{18}+2q^{16}+3q^{15}+5q^{14}+q^{13}+4q^{12}+6q^{11}+10q^{10}+8q^9+8q^8+4q^7+2q^6+q^5+2q^4+q^2+1$.
When $q=2$, there's $A_2(14,4,4)\ge 1259180741$, which exceeds the current best bound 1258757174.

Let $n_1=8,\ n_2=7$, we have $A_q(15,4,4)\ge A_q(8,4,4)\times  q^{21}+q^{21}+2q^{19}+3q^{18}+5q^{17}+q^{16}+4q^{15}+6q^{14}+11q^{13}+10q^{12}+12q^{11}+9q^{10}+8q^9+4q^8+3q^7+2q^6+q^5+q^4+q^3+q^2+2q+1$.
When $q=2$, there's $A_2(15,4,4)\ge 10073479745$, which exceeds the current best bound 10071464646.

Let $n_1=8,\ n_2=8$, we have $A_q(16,4,4)\ge A_q(8,4,4)\times q^{24}+q^{24}+2q^{22}+3q^{21}+5q^{20}+q^{19}+4q^{18}+6q^{17}+11q^{16}+11q^{15}+11q^{14}+11q^{13}+15q^{12}+6q^{11}+5q^{10}+4q^{9}+ 5q^8+q^7+2q^6+q^5+7q^4+q^3+q^2+1$.
When $q=2$, there's $A_2(16,4,4)\ge  80587907742$, while  the current best bound is 80590267742 in the paper \cite{He2019Construction,HeinleinGeneralized,Antonio2019Combining}.

Let $n_1=8,\ n_2=9$, we have $A_q(17,4,4)\ge A_q(8,4,4)\times  q^{27}+
q^{27}+2q^{25}+3q^{24}+5q^{23}+q^{22}+4q^{21}+6q^{20}+11q^{19}+11q^{18}+15q^{17}
+13q^{16}+12q^{15}+8q^{14}+6q^{13}+6q^{12}+7q^{11}+q^{10}+2q^{9}+5q^8+4q^7+q^6+q^4+q$.
When $q=2$, there's $A_2(17,4,4)\ge 644703872849$,  while  the current best bound is 644711939518.

Let $n_1=8,\ n_2=10$, we have $A_q(18,4,4)\ge A_q(8,4,4)\times  q^{30}+q^{30}+2q^{28}+3q^{27}+5q^{26}+q^{25}+4q^{24}+6q^{23}+11q^{22}+11q^{21}+14q^{20}+15q^{19}
+14q^{18}+6q^{17}+5q^{16}+5q^{15}+9q^{14}+2q^{13}+ 6q^{12}+5q^{11}+5q^{10}+2q^9+3q^8
+2q^6+q^5+2q^4+q^2+1$.
When $q=2$, there's $A_2(18,4,4)\ge 5157631206341$,  while  the current best bound is 5157723124262.

Let $n_1=8, n_2=11$, we have  $A_q(19,4,4)\ge A_q(8,4,4)\times  q^{33}+q^{33}+2q^{31}+3q^{30}+5q^{29}+q^{28}+4q^{27}+6q^{26}+11q^{25}+11q^{24}+14q^{23}+13q^{22}+14q^{21}+5q^{20}+6q^{19}
+6q^{18}+8q^{17}+4q^{16}+8q^{15}+8q^{14}+7q^{13}+3q^{12}+3q^{11}+4q^{10}+2q^9+3q^8+2q^7+2q^6+q^5+q^4+q^3+q^2+2q+1$.
When $q=2$, there's $A_2(19,4,4)\ge 41261041141953$, while  the current best bound is 41261547000158.

All the upon improvements are listed and compared in Table~\ref{tab-n-4-4}.

2) $k>d$

Let $n_1=7, n_2=3$, we have  $A_q(10,4,3)\ge A_q(7,4,3)\times  q^{6}+q^{2}+q+1$.
When $q=2$, $A_2(10,4,3)\ge 21319$, while  the current best bound is 21319 \cite{sascha2019note,HeinleinGeneralized}.

Let $n_1=7, n_2=4$, we have  $A_q(11,4,3)\ge A_q(7,4,3)\times  q^{8}+q^{4}+q^3+2q^2+q+1$.
When $q=2$, $A_2(11,4,3)\ge 85283$, while  the current best bound is 85283 \cite{sascha2019note,HeinleinGeneralized}.

Let $n_1=7, n_2=5$, we have  $A_q(12,4,3)\ge A_q(7,4,3)\times  q^{10}+q^{6}+q^5+2q^4+2q^3+2q^2+q+1$.
When $q=2$, $A_2(12,4,3)\ge 341147$, while  the current best bound is 383111 \cite{sascha2019note,HeinleinGeneralized}.
When set $n_1=9, n_2=3, q=2$, we have the same bound 38311.

When $n$ varies from 13 to 16, and $q$ in the set \{2,3,4,5,7,8,9\}, we have the similar bounds to the results in the paper \cite{sascha2019note,HeinleinGeneralized}.

\section{Conclusion}
A construction for constant dimension code is presented in this letter, and new lower bounds of the sizes of constant dimension codes  $A_q(n,d,k)$ are also given.
This construction gives an improved bounds for the linkage construction with echelon-Ferrers.
The results of these lower bounds in \cite{FagangLi2019Construction} are not the best, and our construction generalize the construction.
With the help of the greedy algorithm, we have improved at least the following lower bounds: $A_q(13,4,4), A_q(14,4,4), A_q(15,4,4)$ (listed in Table~\ref{tab-n-4-4}), the expression of these bounds are also given.
All these bounds exceeds the bounds presented in \cite{FagangLi2019Construction}.
Moreover, the identifying vectors underlying the improved bounds are listed in the appendix.


\section*{Appendix}

Here, we list the identifying vectors underlying the improved lower bounds of $A_q(13,4,4), A_q(14,4,4), A_q(15,4,4)$ in Table~\ref{tab-13-4-4},  Table~\ref{tab-14-4-4} and Table~\ref{tab-15-4-4}, respectively.

\begin{table}[!htbp]
\caption{ Identifying Vectors for Construction of $A_q(13,4,4)$}\label{tab-13-4-4}
\begin{tabular}{|c|c|c||c|c|c|}\hline
  & Identifying Vector  &  Dim  &  & Identifying Vector  &  Dim  \\ \hline
  \hline
    1 &  1100000011000 & 15 & 22 &  1000010001001 & 7 \\ \hline
    2 &  0011000011000 & 13 & 23 &  0100100000101 & 7 \\ \hline
    3 &  1010000010100 & 13 & 24 &  0100001010001 & 7 \\ \hline
    4 &  0110000001100 & 12 & 25 &  0000110000110 & 7 \\ \hline
    5 &  0101000010100 & 12 & 26 &  0000100101100 & 7 \\ \hline
    6 &  0110000010010 & 12 & 27 &  0000101001010 & 7 \\ \hline
    7 &  1100000000110 & 11 & 28 &  0000100110010 & 7 \\ \hline
    8 &  0000110011000 & 11 & 29 &  0001010000101 & 6 \\ \hline
    9 &  1010000001010 & 11 & 30 &  0010100000011 & 6 \\ \hline
    10 &  1001000010010 & 11 & 31 &  0001000110001 & 6 \\ \hline
    11 &  1001000001100 & 11 & 32 &  0010001001001 & 6 \\ \hline
    12 &  0101000001010 & 10 & 33 &  0000010101010 & 6 \\ \hline
    13 &  1000100010001 & 9 & 34 &  1000001000101 & 5 \\ \hline
    14 &  0011000000110 & 9 & 35 &  0100010000011 & 5 \\ \hline
    15 &  0000101010100 & 9 & 36 &  0100000101001 & 5 \\ \hline
    16 &  0001100001001 & 8 & 37 &  0001001000011 & 4 \\ \hline
    17 &  0000001111000 & 8 & 38 &  0000001100110 & 4 \\ \hline
    18 &  0000011001100 & 8 & 39 &  0010000100101 & 4 \\ \hline
    19 &  0010010010001 & 8 & 40 &  1000000100011 & 3 \\ \hline
    20 &  0000011010010 & 8 & 41 &  0000000011110 & 0 \\ \hline
    21 &  0000010110100 & 8 &  &  &  \\ \hline
\end{tabular}
\end{table}

\begin{table}[!htbp]
\caption{ Identifying Vectors for Construction of $A_q(14,4,4)$}\label{tab-14-4-4}
\begin{tabular}{|c|c|c||c|c|c|}\hline
  & Identifying Vector  &  Dim  &  & Identifying Vector  &  Dim  \\ \hline
  \hline
    1 &  11000000110000 & 18 & 31 &  00001100001100 & 10 \\ \hline
    2 &  00110000110000 & 16 & 32 &  00000110010100 & 10 \\ \hline
    3 &  01100000101000 & 16 & 33 &  10000100010001 & 9 \\ \hline
    4 &  10010000101000 & 15 & 34 &  10000100001010 & 9 \\ \hline
    5 &  10100000100100 & 15 & 35 &  00100010100001 & 9 \\ \hline
    6 &  10100000011000 & 15 & 36 &  00100010010010 & 9 \\ \hline
    7 &  11000000001100 & 14 & 37 &  00101000001001 & 9 \\ \hline
    8 &  01010000100100 & 14 & 38 &  00001001010100 & 9 \\ \hline
    9 &  01010000011000 & 14 & 39 &  10000001100010 & 9 \\ \hline
    10 &  00001100110000 & 14 & 40 &  01001000000110 & 9 \\ \hline
    11 &  01100000010100 & 14 & 41 &  00010010010001 & 8 \\ \hline
    12 &  10010000010100 & 13 & 42 &  00010001100001 & 8 \\ \hline
    13 &  00101000100010 & 12 & 43 &  00010001010010 & 8 \\ \hline
    14 &  00000110101000 & 12 & 44 &  00110000000011 & 8 \\ \hline
    15 &  00000011110000 & 12 & 45 &  00100100000110 & 8 \\ \hline
    16 &  00110000001100 & 12 & 46 &  00000011001100 & 8 \\ \hline
    17 &  00010100100010 & 11 & 47 &  00011000000101 & 8 \\ \hline
    18 &  10001000100001 & 11 & 48 &  00010100001001 & 8 \\ \hline
    19 &  10001000010010 & 11 & 49 &  00100001010001 & 7 \\ \hline
    20 &  00001010011000 & 11 & 50 &  01000001001010 & 7 \\ \hline
    21 &  00001001101000 & 11 & 51 &  10000010000110 & 7 \\ \hline
    22 &  00001010100100 & 11 & 52 &  10000010001001 & 7 \\ \hline
    23 &  11000000000011 & 10 & 53 &  00001100000011 & 6 \\ \hline
    24 &  01000010100010 & 10 & 54 &  01000010000101 & 6 \\ \hline
    25 &  00000101100100 & 10 & 55 &  10000001000101 & 5 \\ \hline
    26 &  00000101011000 & 10 & 56 &  00000011000011 & 4 \\ \hline
    27 &  01001000010001 & 10 & 57 &  00000000111100 & 4 \\ \hline
    28 &  01000100100001 & 10 & 58 &  00000000110011 & 2 \\ \hline
    29 &  01000100010010 & 10 & 59 &  00000000001111 & 0 \\ \hline
    30 &  00011000001010 & 10 &  &  &  \\ \hline
\end{tabular}
\end{table}

\begin{table}[!htbp]
\caption{ Identifying Vectors for Construction of $A_q(15,4,4)$}\label{tab-15-4-4}
\begin{tabular}{|c|c|c||c|c|c|}\hline
  & Identifying Vector  &  Dim  &  & Identifying Vector  &  Dim  \\ \hline
  \hline
    1 &  110000001100000 & 21 & 45 &  010000101000001 & 11 \\ \hline
    2 &  001100001100000 & 19 & 46 &  010000100010100 & 11 \\ \hline
    3 &  101000001010000 & 19 & 47 &  000101000100001 & 11 \\ \hline
    4 &  011000001001000 & 18 & 48 &  101000000000101 & 11 \\ \hline
    5 &  011000000110000 & 18 & 49 &  001100000000110 & 11 \\ \hline
    6 &  010100001010000 & 18 & 50 &  100001000010010 & 11 \\ \hline
    7 &  100100001001000 & 17 & 51 &  100001000001100 & 11 \\ \hline
    8 &  100100000110000 & 17 & 52 &  100000100100010 & 11 \\ \hline
    9 &  110000000011000 & 17 & 53 &  000000110011000 & 11 \\ \hline
    10 &  101000000101000 & 17 & 54 &  100000010100100 & 11 \\ \hline
    11 &  000011001100000 & 17 & 55 &  001010000001010 & 11 \\ \hline
    12 &  010100000101000 & 16 & 56 &  000101000001010 & 10 \\ \hline
    13 &  010010001000100 & 15 & 57 &  000100010010100 & 10 \\ \hline
    14 &  000001101010000 & 15 & 58 &  010100000000101 & 10 \\ \hline
    15 &  000000111100000 & 15 & 59 &  000100010100010 & 10 \\ \hline
    16 &  001100000011000 & 15 & 60 &  011000000000011 & 10 \\ \hline
    17 &  000101001000100 & 14 & 61 &  001000100100001 & 10 \\ \hline
    18 &  000010011010000 & 14 & 62 &  000100011000001 & 10 \\ \hline
    19 &  000010101001000 & 14 & 63 &  001000100010010 & 10 \\ \hline
    20 &  001010000100100 & 14 & 64 &  001000100001100 & 10 \\ \hline
    21 &  100010001000010 & 14 & 65 &  000100100010001 & 9 \\ \hline
    22 &  000010100110000 & 14 & 66 &  000011000000110 & 9 \\ \hline
    23 &  110000000000110 & 13 & 67 &  100100000000011 & 9 \\ \hline
    24 &  010001001000010 & 13 & 68 &  010001000001001 & 9 \\ \hline
    25 &  000001100101000 & 13 & 69 &  010000100001010 & 9 \\ \hline
    26 &  000001011001000 & 13 & 70 &  010000010100001 & 9 \\ \hline
    27 &  000001010110000 & 13 & 71 &  010000010010010 & 9 \\ \hline
    28 &  100010000010100 & 13 & 72 &  010000010001100 & 9 \\ \hline
    29 &  010001000100100 & 13 & 73 &  100000010001010 & 8 \\ \hline
    30 &  100000101000100 & 13 & 74 &  100000010010001 & 8 \\ \hline
    31 &  010010000100010 & 13 & 75 &  100000100001001 & 8 \\ \hline
    32 &  001010001000001 & 13 & 76 &  000000001111000 & 8 \\ \hline
    33 &  000011000011000 & 13 & 77 &  000010100000101 & 7 \\ \hline
    34 &  100010000100001 & 12 & 78 &  000000110000110 & 7 \\ \hline
    35 &  001001000100010 & 12 & 79 &  001000010001001 & 7 \\ \hline
    36 &  000010010101000 & 12 & 80 &  000001010000101 & 6 \\ \hline
    37 &  000100101000010 & 12 & 81 &  000001100000011 & 6 \\ \hline
    38 &  001001000010100 & 12 & 82 &  000010010000011 & 5 \\ \hline
    39 &  000110000010010 & 12 & 83 &  000000001100101 & 4 \\ \hline
    40 &  000110000001100 & 12 & 84 &  000000001010110 & 3 \\ \hline
    41 &  000100100100100 & 12 & 85 &  000000000110011 & 2 \\ \hline
    42 &  100001001000001 & 12 & 86 &  000000001001011 & 1 \\ \hline
    43 &  001000011000100 & 12 & 87 &  000000000101110 & 1 \\ \hline
    44 &  010010000010001 & 11 & 88 &  000000000011101 & 0 \\ \hline
\end{tabular}
\end{table}

\end{document}